\documentclass{amsart}

\usepackage[utf8]{inputenc}
\usepackage{hyperref}
\usepackage{amsthm}
\usepackage{amsmath}
\usepackage{amssymb}
\usepackage{graphicx}
\usepackage{float}

\theoremstyle{definition}
\theoremstyle{theorem}
\newtheorem{theorem}{Theorem}[section]
	
\numberwithin{equation}{section}

\newtheorem{lemma}[theorem]{Lemma}

\newtheorem{proposition}[theorem]{Proposition}

\newtheorem*{remark*}{Remark}

\newtheorem{remark}[theorem]{Remark}

\DeclareMathOperator{\tr}{Tr}

\numberwithin{equation}{section}

\begin{document}

\title[WKBJ method in a RD context]{WKBJ approximation for linearly coupled systems: asymptotics of reaction-diffusion systems}


\author{Juraj Kov{\'a}{\v c}}
\address{Dept of Mathematics, FNSPE, Czech Technical University in Prague, Czech Republic}
\email{kovacjur@fjfi.cvut.cz}

\author{V{\'a}clav Klika}
\address{Dept of Mathematics, FNSPE, Czech Technical University in Prague, Czech Republic}
\email{vaclav.klika@fjfi.cvut.cz}



\date{\today}


\begin{abstract}
  Asymptotic analysis has become a common approach in investigations of reaction-diffusion equations and pattern formation, especially when considering generalizations to the original model, such as spatial heterogeneity, where finding an analytic solution even to the linearized equations is generally not possible. The WKBJ method, one of the more robust asymptotic approaches for investigating dissipative phenomena captured by linear equations, has recently been applied to the Turing model in a heterogeneous environment. It demonstrated the anticipated modifications to the results obtained in a homogeneous setting, such as localized patterns and local Turing conditions. In this context, we attempt a generalization of the scalar WKBJ theory to multicomponent systems. Our broader mathematical approach results in general approximation theorems for systems of ODEs. We discuss the cases of exponential and oscillatory behaviour first before treating the general case. Subsequently, we demonstrate the spectral properties utilized in the approximation theorems for a typical Turing system, hence suggesting that such an approximation is reasonable. Note that our line of approach is via showing that the solution is close (using suitable weight functions for measuring the error) to a linear combination of Airy functions.
\end{abstract}

\maketitle

\section{Introduction}

The simple observation by Turing \cite{Turing} that reaction and diffusion can account for a large variety of pattern formation phenomena has given rise to a vast amount of research in the past decades, both theoretical and experimental. The model based on the transition from homogenenity to heterogenenity due to the effects of diffusion is just as strikingly counterintuitive as it is remarkably simple. The broad applicability ranging from developmental biology and embryogenesis \cite{Maini1397} to distribution and dispersal of animal populations \cite{Murray03} has stimulated interdisciplinary efforts to further understand and develop the model.  

 Rejected by large portions of the scientific community at first \cite{Waddington}, partly due to its austerity, the Turing model has been the subject of advanced analyses and attempted generalizations seeking to enhance its complexity, relax its constraints or modify its properties, and to challenge Turing's seemingly simplistic original approach to pattern formation. All these novel concepts and approaches require mathematical tools beyond the classical scope summarized comprehensively by Cross \& Hohenberg \cite{Cross+H}. These tools may include, but are not restricted to, asymptotic analysis \cite{white2010}, weakly non-linear analysis \cite{Wollkind}, theory of partial differential equations \cite{PDEs} etc. Other approaches, such as analyzing pseudospectra and non-normality, have been shown to be of limited significance in the context of pattern formation and the Turing model \cite{non-normality}, while domain growth has been observed to affect pattern selection in almost every way possible \cite{madzvamuse2010stability,Krause-growth, van2021turing}. Note, however, that there are other means to obtain self-organization (some particularly relevant in the context of developmental biology) than via Turing's instability. For example, domain-size driven instability \cite{advection}, which coincides with Turing's diffusion-driven instability in reaction-diffusion systems with zero flux boundary conditions, disputes the necessity of short-range activation and long-range inhibition, one of the cornerstone phenomena of the classical Turing mechanism. 

Two of the aforementioned concepts of generalization will be of key importance in this paper, namely heterogeneity and asymptotic analysis. The effect of spatial heterogeneity has been a long-standing problem with well-understood implications limited to particular reaction kinetics and shadow systems \cite{ward2002dynamics} and with the potential to extend the Turing (parameter) space and incite localized patterns \cite{Page}. Recently, under the assumption of time-scale separation (with diffusion coefficients small compared to reaction terms) or small spatial variation, it has been formally shown that diffusion-driven instability conditions can be perceived locally \cite{Klika,kozak2019pattern}. Additionally, heterogeneity has been shown to result in complex and unexpected behaviour \cite{krause2018heterogeneity}. Hence, unsurprisingly, it can be concluded that relying on linear stability analysis does not capture the full complexity of pattern formation.

Motivated by the WKBJ analysis of a two-component Turing system by Krause et al.\cite{Klika}, our goal here is to justify the deployment of these tools in a reaction-diffusion (RD) setting. It should be noted that despite being far more well-established in various fields of quantum mechanics, WKBJ analysis has been applied to a RD system in an infinite 1-dimensional domain by Dewel \& Borckmans \cite{DEWEL}. However, in line with the analysis presented in \cite{Klika}, our focus is on a finite domain of given length. We first approach the issue from a broader mathematical perspective and offer a generalization of a WKBJ approximation theorem to systems of linearly coupled ordinary differential equations in the usual WKBJ setting including an infinitesimal parameter \cite{Bender}. For the sake of lucidity, we distinguish carefully between the case of exponential behaviour and the case of oscillatory modes. The classical WKBJ modes are not directly represented in the approximation. In the former case, they instead appear in the proof as weight functions with asymptotic behaviour corresponding to the approximate solution represented in terms of modified Bessel functions. This remains true for the latter case, although the use of weight functions is not necessary there. In both cases, the proof is to a large extent based on the spectral properties of the system in question. Along with adding rigour to the WKBJ analysis, we provide estimates for the error of representing the solution to the coupled problem by Airy functions, the solutions to a scalar reaction-diffusion equation with linear spatial dependence.

Next, using these results, we demonstrate that a typical reaction-diffusion system considered in the limit of large growth rates has all the spectral properties utilized in the proof of the validity of the asymptotic approximation. We then offer a discussion of the result, its implications, and possible generalizations before demonstrating the most technical steps in the proof of the approximation theorem of section \ref{system} in Appendix \ref{details}. Appendix \ref{App Turing} summarizes the basic concepts of RD equations and Turing instability, and derives some of its properties and constraints useful in the text for a homogeneous setting, commenting on their generalizations to spatially heterogeneous environments. Should the reader not be familiar with the basic concepts and terminology of RD equations, Appendix \ref{App Turing} is recommended as a starting point to this article, especially helpful in comprehending section \ref{sect_WKBJ RD}.

\section{WKBJ approximation for a system of ODEs} \label{system}	
The WKBJ method is a powerful set of tools for obtaining approximate solutions to linear ordinary differential equations. Being an asymptotic method, it usually assumes that the highest-order derivative in the equation is multiplied by a parameter, say $\varepsilon$, and offers an asymptotic expansion of the solution for the limit $\varepsilon \rightarrow 0$. In section \ref{sect_WKBJ RD} we will also assume a different limit as well as a different position of the asymptotic parameter in the equation. We will see, however, that for the Schrödinger equation, which will be relevant for our considerations, these two approaches are analogous under additional assumptions, and the transition between them is simply a matter of rescaling. Without much further ado, let us note for future reference that for the one-dimensional Schrödinger equation 
	\begin{equation} \label{SR}
	\varepsilon^2 y'' + Q(x) y = 0,
	\end{equation}
where we assume $Q(x) \neq 0$ and denote the infinitesimal parameter $\varepsilon^2$ rather than $\varepsilon$ for greater convenience, the first-order WKBJ approximation has the form
	\begin{equation}\label{1st-order WKBJ}
	y(x) \sim c_1 Q(x)^{-1/4} \exp \left( \frac{1}{\varepsilon} \int^x_a \sqrt{-Q(s)} ds \right) + c_2 Q(x)^{-1/4} \exp \left(- \frac{1}{\varepsilon} \int^x_a \sqrt{-Q(s)} ds \right).
	\end{equation}
We refer the reader to external sources such as chapter 10 in \cite{Bender} for a more extensive and detailed account of the scalar WKBJ theory.

We are now going to offer an approximation theorem for the WKBJ approximation applied to a system of ordinary differential equations. However, we will not proceed in the usual manner, first listing the assumptions and then offering the proof. Instead, we will write down the ideas in their natural order, trying to demonstrate how and where different conditions and assumptions arise, thus proving in advance a theorem that will then summarize the entire procedure. Subsequently, we will relate the key properties and assumptions of the proof to a reaction-diffusion setting in the next section.

Our focus is on a multicomponent Schrödinger equation of the form 
\begin{equation} \label{WKBepsilon}
\varepsilon^2 \mathbf{y}'' + \mathbb{Q}(x) \mathbf{y} = 0.
\end{equation}
In what follows, we will largely rely on the spectral properties of ${\mathbb{Q}}$. Therefore, we note that for $$\tilde{\mathbb{Q}}(x;\varepsilon) = \frac{1}{\varepsilon^2} \mathbb{Q}(x)$$ we have $\sigma (\tilde{\mathbb{Q}}) = \frac{1}{\varepsilon^2} \sigma (\mathbb{Q})$, with $\sigma$ denoting the spectrum of the matrix. It is thus clear that the limit $\varepsilon \rightarrow 0$ will render the spectrum of $\tilde{\mathbb{Q}}(x)$ unbounded. We shall also assume the eigenvalues $\mu_j(x)$ sufficiently smooth (as their derivatives will appear in the analysis below). Soon we will see that it is the spectral rescaling $\sigma (\tilde{\mathbb{Q}}) = \frac{1}{\varepsilon^2} \sigma (\mathbb{Q})$ that accounts for the $1/\varepsilon$-factor in the phase of \eqref{1st-order WKBJ}. Motivated by RD equations (see Appendix \ref{App Turing}), we will first assume that the spectrum $\sigma(\mathbb{Q}(x))$ of $\mathbb{Q}(x)$ is negative.\footnote{The spectral properties of RD equations in the limit of large growth rates are discussed in section \ref{sect_WKBJ RD}.} Nevertheless, Krause et al. \cite{Klika} consider the case of positive eigenvalues (for large but finite growth rates that satisfy the conditions for Turing instability), making the results of section \ref{sect_oscillatory} of relevance for their analysis. 

Using the above observations along the way, our considerations will hereinafter refer to the general problem
\begin{equation}\label{SR 2D}
  \mathbf{y}''(x) + \tilde{\mathbb{Q}}(x; \varepsilon) \mathbf{y}(x) = 0.
\end{equation}

\subsection{Exponential case}\label{sect_exponential}
First, assume that the spectrum $\sigma(\mathbb{Q}(x))$ of $\mathbb{Q}(x)$ is real and negative for every $x \in [0,L]$. The eigenvalues $\mu(x)$\footnote{To simplify the notation we choose not to index the eigenvalues. Just note that all the quantities introduced below depend on the choice of the eigenvalue $\mu(x)$, unless otherwise stated.} of $\tilde{\mathbb{Q}}(x)$ will hence satisfy 
\begin{equation*}
 \mu(x) { \rightarrow} - \infty \text{ for } {\varepsilon \rightarrow 0}, \forall x \in [0,L].
 \end{equation*}
 The procedure will be as follows: we will find the exact solutions $u,v$ for the specific case ${\mathbb{Q}}(x) = -x \mathbb{I}$, rewrite them in terms of $\mathbb{Q}$ (more precisely, in terms of the eigenvalues $\mu(x)$) rather than $x$ and find that for a general ${\mathbb{Q}}(x)$ they do not satisfy \eqref{SR 2D}, but instead obtain additional linear correction terms. We will then be able to use these terms to find an implicit integral form of the solution and apply the limit $\varepsilon \rightarrow 0$ to show that the correction of the solution vanishes in this limit.

We start by defining
\begin{equation}\label{xi}
\xi(x) \equiv \int_0^x \sqrt{-\mu(t)} dt,
\end{equation}
a quantity that is positive and becomes unbounded with $\varepsilon$ approaching $0$ for every $x>0$. Keeping these properties in mind for later use, we can now define the two functions 
\begin{equation}\label{u and v}
\begin{aligned}
u(x) & \equiv \sqrt{\frac{\xi(x)}{\xi'(x)}} K_{1/3} \left( \xi(x) \right), \\
v(x) & \equiv \sqrt{\frac{\xi(x)}{\xi'(x)}} I_{1/3} \left( \xi(x) \right),
	\end{aligned}
\end{equation}
with $K_{\alpha}(x), I_{\alpha}(x)$ being the modified Bessel functions defined as the two linearly independent solutions to $x^2 y'' + x y' - (x^2 + \alpha^2)y = 0$. As stated above, for ${\mathbb{Q}}(x) = -x \mathbb{I}$ we have two independent Airy equations and $u,v$ are their exact solutions. Even more importantly, using the recurrent relations
\begin{equation}\label{recurrent}
	\begin{aligned}
K_\alpha'(x) &= \frac{\alpha}{x}K_\alpha(x) - K_{\alpha+1}(x), \\
I_\alpha'(x) &= \frac{\alpha}{x}I_\alpha(x) + I_{\alpha+1}(x), \\
	& = I_{\alpha-1}(x) - \frac{\alpha}{x}I_\alpha(x),
	\end{aligned}
\end{equation}
along with the symmetric property $K_{-\alpha}(x) = K_\alpha(x)$ (cf. Theorems 4.15 and 4.16 in \cite{Bell1968}) and after some exhausting algebra, we arrive at the following property (see Lemma \ref{T proof} in Appendix \ref{details}) valid for both $\psi=u$ and $\psi=v$:
\begin{equation}\label{2nd der}
\psi'' = \left[ (\xi')^2 - T \right] \psi = \left[ -\mu - T \right] \psi,
\end{equation}
where
\begin{equation} \label{T}
	\begin{aligned}
T(x) &= \frac{1}{4} \left[ \frac{5}{9} \left( \frac{\xi'}{\xi} \right)^2 + 2 \frac{\xi' \xi'''}{\xi'^2} - 3 \left( \frac{ \xi'' }{\xi'} \right)^2 \right](x) = \frac{1}{4} \left[ \frac{5}{9} \left( \frac{\xi'}{\xi} \right)^2 + \frac{(\xi'^2)''}{\xi'^2} - \frac{5}{4} \frac{\left((\xi'^2)'\right)^2 }{\xi'^4} \right](x) = \\
	& = \frac{1}{4} \left[ -\frac{5}{9} \frac{\mu}{\xi^2} + \frac{\mu''}{\mu} - \frac{5}{4} \frac{\mu'^2 }{\mu^2} \right](x).
	\end{aligned}
\end{equation}
Note that the powers of $\xi$ (or $\mu$) in the numerator and in the denominator are identical for all terms. Due to the scaling of the spectrum of $\tilde{\mathbb{Q}}(x)$ we have a separation of the dependence of the eigenvalue $\mu$ on $x$ and $\varepsilon$, namely $\mu(x)=\frac{1}{\varepsilon^2} m(x)$. As a result, the derivatives of the eigenvalues $\mu$  are all of the same order in $\varepsilon$ and  the same applies to $\xi$ and its derivatives. Hence,  regardless of the choice of the eigenvalue $\mu,$ the function $T(x)$ is $\varepsilon$-independent, which also means bounded in the limit $\varepsilon \rightarrow 0$, a fact we will make good use of later on.

Observe now that $u$ and $v$ are the two linearly independent solutions to the same linear equation \eqref{2nd der}, with $K_{1/3}(x)$ regular at $x \rightarrow +\infty$ and $I_{1/3}(x)$ regular for $x \rightarrow 0$, and their Wronskian $W[u,v](x)$ can be shown to be constant from Sturm-Liouville theory. In addition, $u$ is the exponentially decaying solution while $v$ is the exponentially growing solution and $W[u,v](x)$ is equal to $1$. \footnote{For a more detailed account of the properties of the modified Bessel functions and Airy functions, we refer the reader to the aforementioned book by Bell or the well-known book of special functions by Abramowitz and Stegun, especially chapters 9.6 and 10.4.\cite{Abramowitz}}

We can now proceed by defining another function $G(x,s) = u(x) v(s) - v(x) u(s)$.\footnote{This is obviously not a proper Green's function. However, the property that we are about to show 'justifies' this notation.} This function can now be used as a part of the kernel of an integral operator to find the exact solution to \eqref{SR 2D}. More precisely, it is just a matter of direct differentiation and applying the property \eqref{2nd der} of $u$ and $v$ to show that the implicitly defined function \footnote{We omit $v$ from the non-integral part of $y$ since $\xi$ is unbounded for any $x>0$ when $\varepsilon \rightarrow 0$ and $v$, in turn, becomes unbounded for an unbounded argument.}
\begin{equation}\label{int solution}
y(x) = u(x) + \int_x^L G(x,s) T(s)y(s)ds 
\end{equation}
satisfies 
\begin{equation*}
y''(x) = \left[ -T(x) + \xi'^2(x) \right] y(x) + T(x)y(x)W[u,v](x) = \xi'^2 y(x).
\end{equation*}

In the next step, we would like to use these observations concerning the scalar WKBJ method and connect it to the vector case. To this end, we shall use the eigenvectors $\mathbf{p}$ of $\tilde{\mathbb{Q}},$ defined by $\tilde{\mathbb{Q}}(x)\mathbf{p} = \mu(x) \mathbf{p},$ to diagonalize $\tilde{\mathbb{Q}}$ in order to obtain a solution to \eqref{SR 2D}. Note that assuming $\tilde{\mathbb{Q}}$ diagonalizable on the whole interval $[0,L]$ is a rather mild restriction as almost every matrix is diagonalizable. Moreover, we believe that this technical assumption can be removed altogether.

Due to the spatial dependence of the eigenvector $\mathbf{p}(x)$ we shall proceed with greater care. 
Consider the rotated scalar solutions 
\begin{equation}\label{solution SR 2D}
\mathbf{y}(x) = y(x) \mathbf{p}(x).
\end{equation}
Requiring $\mathbf{y}(x)$ to be a solution to \eqref{SR 2D} yields
\begin{equation*}
  (y'' \mathbb{I}+y \tilde{\mathbb{Q}}) \mathbf{p} = -2 y' \mathbf{p}'-y \mathbf{p}''
\end{equation*}
or, using the above property of $y$,
\begin{equation}\label{vector eq}
0 = y (-\mu \mathbb{I}+ \tilde{\mathbb{Q}}) \mathbf{p} = -2 y' \mathbf{p}'-y \mathbf{p}''.
\end{equation}
Denoting $\overline{\mathbf{p}}$ the unit eigenvector, we substitute $\mathbf{p}(x) = \beta(x)\overline{\mathbf{p}}(x)$, where $\beta(x)$ for now represents an undetermined amplitude. Equation \eqref{vector eq} can then be rewritten as
\begin{equation*}
  (-\mu \mathbb{I}+ \tilde{\mathbb{Q}}) \overline{\mathbf{p}} =-\frac{1}{\beta y}\left[\beta'' y\overline{\mathbf{p}}+2\beta'(y\overline{\mathbf{p}})'+\beta(2 y' \overline{\mathbf{p}}'+y\overline{\mathbf{p}}'')\right].
\end{equation*}
As the matrix on the left-hand side is singular and independent of the yet undetermined amplitude function $\beta$, the Fredholm alternative imposes a solvability condition. Let us denote by $\mathbf{s}^T$ the left eigenvector \footnote{The dimensions of the left and right eigenspace corresponding to the same eigenvalue are always equal. For simplicity, here we assume these eigenspaces one-dimensional.} of $(-\mu \mathbb{I}+ \tilde{\mathbb{Q}})$ corresponding to the zero eigenvalue. The solvability condition thus becomes
\begin{equation} \label{eq:Amplitude}
  0=\beta'' y\mathbf{s}.\overline{\mathbf{p}}+2\beta' \mathbf{s}.(y\overline{\mathbf{p}})'+\beta[2 y' \mathbf{s}.\overline{\mathbf{p}}'+y\mathbf{s}.\overline{\mathbf{p}}''],
\end{equation}
which is a (scalar) second-order linear ordinary differential equation for the amplitude $\beta$. Crucially, there exists an amplitude $\beta$ that solves the differential equation \eqref{eq:Amplitude}, and hence there exists $\beta$ such that $0=2 y'\mathbf{p}'+y \mathbf{p}''$.
Therefore, the rotated scalar solutions \eqref{solution SR 2D} with this particular amplitude satisfy \eqref{SR 2D} as we have
\begin{equation*}
  \mathbf{y}'' + \tilde{\mathbb{Q}}(x) \mathbf{y} = (y'' \mathbb{I}+y \tilde{\mathbb{Q}}) \mathbf{p} + 2 y' \mathbf{p}'+y \mathbf{p}''= (y'' \mathbb{I}+y \tilde{\mathbb{Q}}) \mathbf{p}=
   \xi'^2 y(x) \mathbf{p} +  \mu(x) y(x) \mathbf{p} = 0.
\end{equation*}
Note that spatially constant eigenvectors $\mathbf{p}$ present a simple special case where the equation for the amplitude $\beta$ has a constant solution. A sufficient condition for such a situation is, for example, $\tilde{\mathbb{Q}}$ having the form $\tilde{\mathbb{Q}}=q(x) \mathbb{M}$.

Next, we would like to show that the correction to $u$ given by the integral in \eqref{int solution} becomes arbitrarily small for sufficiently small $\varepsilon$. Hence, the particular knowledge of solutions to the scalar Airy differential equation can be used to construct arbitrarily precise solutions to the general vector problem once the parameter $\varepsilon$ is sufficiently small. However, this requires some subtle handling.

Before we dig deeper into that, we take a small detour that will bring us closer to completing the proof. As things stand, the integral in \eqref{int solution} features the unknown function $y$, which is very inconvenient but can be taken care of in a rather elegant manner. Let us introduce the integral operator 
\begin{equation*}
(A f)(x) \equiv \int_x^L G(x,s) T(s)f(s)ds.
\end{equation*}
Then \eqref{int solution} can be rewritten $(\mathbb{I} - A)y = u$. If we could prove that $\Vert A \Vert < 1$, we would be able to define the inverse to the operator on the left-hand side via the Neumann series. Let us for now assume that this can be achieved for a suitably small $\varepsilon$, a property we shall show below. Then we would have $y = (\mathbb{I} - A)^{-1}u$ and the correction to $u$ in $y$ could thus be written as
\begin{equation*}
y - u = (\mathbb{I} - A)^{-1}u - u = \left[ (\mathbb{I} - A)^{-1} - \mathbb{I} \right] u = \left( \sum_{n=1}^\infty A^n \right) u,
\end{equation*}
which implies
\begin{equation*}
\Vert y-u \Vert_{L^\infty(\Omega)} \leq \sum_{n=1}^\infty \Vert A \Vert^n_\infty \Vert u \Vert_{L^\infty(\Omega)} = \frac{\Vert A \Vert_\infty}{1 - \Vert A \Vert_\infty}\Vert u \Vert_{L^\infty(\Omega)}.
\end{equation*}
In this estimate, we denoted by $\Vert \cdot \Vert_\infty $ the operator norm induced by $\Vert \cdot \Vert_{L^\infty(\Omega)} $ and set $\Omega = (0,L)$. In this way, we eliminated the unknown function $y$ from our estimate of the correction $r(x) \equiv y(x) - u(x)$. Note that both norms on the right-hand side generally depend on $\varepsilon$ (although $u$ is a bounded function on $\mathbb{R}_+$ so a global, $\varepsilon$-independent upper bound can be found for its norm). If we could prove that $\Vert A \Vert_\infty$ is not only finite but vanishes in the limit $\varepsilon \rightarrow 0$, the proof would be complete.

To show that, we restrict the discussion of the estimates below only to $s\geq x$ as follows from the definition of $A$.
As stated earlier, $u$ is the decaying and $v$ the increasing solution, both exponentially. This fact can be expressed more precisely, namely in terms of the suitable weight functions 
\begin{equation}\label{weights}
w_1(x) \equiv (-\mu(x))^{1/4}, \quad w_2(x) \equiv \exp \left( \int_0^x \sqrt{-\mu(s)}ds \right).
\end{equation}
First, observe that both $u(x)$ and $v(x)$ are bounded near $x=0$ due to the pre-factor $\left( \xi/\xi' \right)^{1/2}$ and by the asymptotic relations utilized in the proof of Lemma \ref{integral}. Combining this observation with the asymptotic relations for large arguments 9.7.1 and 9.7.2 in \cite{Abramowitz} (for $z=\xi$), we obtain for the functions $u,v$ the upper bounds
\begin{equation}\label{upper bounds}
\vert u(x) \vert \leq \frac{M}{w_1(x) w_2(x)}, \qquad \vert v(x) \vert \leq \frac{M w_2(x)}{w_1(x)},
\end{equation}
which imply $\vert u(x)v(x) \vert \leq \frac{M^2}{w_1(x)^2}= \frac{M^2}{\sqrt{-\mu(x)}}$. These estimates are global, i.e. $M$ can be chosen independent of the particular choice of the eigenvalue and of the value of $\varepsilon$. Note that the functions $\nu_1(x) \equiv \frac{1}{w_1(x) w_2(x)}$ and $\nu_2(x) \equiv \frac{w_2(x)}{w_1(x)}$ correspond to the two modes in \eqref{1st-order WKBJ}\footnote{This is the only reason we include $w_1$ in the proof. Since the eigenvalues $\mu(x)$ are assumed nonzero, both $w_1$ and $1/w_1$ are bounded. Consequently, the validity of the subsequent estimates would be preserved even without $w_1$.}, either both harmonic or one exponentially growing and the other one exponentially decaying, with the latter being the case under our assumptions.

We are now ready to estimate $G(x,s)$ within the integral of the operator $A$. First of all, note that both weight functions $w_1, w_2$ are bounded in $x$ (for a fixed $\varepsilon$) as, from our initial assumptions, $\mu<0$. Secondly, $\xi(t)$ is increasing in $t$, and hence the second term, $v(x)u(s)$, in $G(x,s)$ can be shown to be bounded. However, we cannot generally show the boundedness of the first term, $u(x)v(s)$. Therefore, we introduce one more definition, namely a change in the measure of the residue $y-u$.

Introducing the notation $\psi_{12}(x) \equiv w_1(x) w_2(x) \psi(x)$, we rewrite \eqref{int solution} as 
\begin{equation} \label{bounded}
y_{12}(x) = u_{12}(x) + \int_x^L T(s) \frac{ u_{12}(x) v_{12}(s) - v_{12}(x) u_{12}(s)}{w_2^2(s)} \frac{1}{w_1^2(s)} y_{12}(s)ds \equiv u_{12} +A_{12}y_{12},
\end{equation}
where we introduced the operator $A_{12}=w_1(x)w_2(x) A$. In analogy to the discussion of the properties of $A$ above, we aim to show that $\Vert A_{12}\Vert\rightarrow 0$ as $\varepsilon \rightarrow 0$, which would guarantee that the residual vanishes (if measured in $L^\infty(\Omega)$ with the weight $w_1(x) w_2(x)$):
\begin{equation*}
\Vert y_{12}-u_{12} \Vert_{L^\infty(\Omega)} \leq \frac{\Vert A_{12} \Vert_\infty}{1 - \Vert A_{12} \Vert_\infty}\Vert u_{12} \Vert_{L^\infty(\Omega)}\leq M \frac{\Vert A_{12} \Vert_\infty}{1 - \Vert A_{12} \Vert_\infty},
\end{equation*}
thanks to the upper bound \eqref{upper bounds} on $u(x)$.

For the second fraction in \eqref{bounded}  we have 
\begin{equation}\label{estimate1}
\left| \frac{v_{12}(x) u_{12}(s)}{w_2^2(s)} \right| \leq \left| u(x) v(x) \frac{w_1(s)w_2(x)w_1(x)}{w_2(s)} \right| \leq M^2 \left| \frac{w_1(s)w_2(x)}{w_1(x)w_2(s)} \right| \leq M^2 \left| \frac{w_1(s)}{w_1(x)} \right|,
\end{equation}
which is bounded by the boundedness of $\mu$ with respect to $x$. In addition, the ratio $\frac{w_1(s)}{w_1(x)}$ is $\varepsilon$-independent as follows from the scaling argument of the spectrum of $\tilde{Q}$. Similarly, using \eqref{upper bounds} and the fact that $\left| \frac{v(s)}{w_2(s)} \right| \leq \frac{M}{\vert w_1(s) \vert}$, for the first term in \eqref{bounded} we obtain 
\begin{equation}\label{estimate2}
\left| \frac{u_{12}(x) v_{12}(s)}{w_2^2(s)} \right| \leq M \left| w_2(x)w_1(x)u(x) \right| \leq M^2.
\end{equation}
Combining \eqref{estimate1} and \eqref{estimate2} hence yields a uniform estimate of $(u_{12}(x)v_{12}(s)-v_{12}(x)u_{12}(s))/w_2^2(s)$ in both $x$ and $\varepsilon$.

Since $L$ is finite and $1/w_1^2$ is bounded from above in both $x$ and $\varepsilon$ from its definition and the assumption $\mu<0$, all that is left to show are the properties of $T$. There are two main concerns in this regard. We have argued above that $T(x)$ is bounded with respect to $\varepsilon$ but we have yet to explore its behaviour with respect to $x$. As $\mu<0$, the only potential singularity may originate from the first term, $\left( \xi'(s)\right)^2/ (\xi(s))^2 $, in particular for $x\rightarrow 0_+$ and $s \rightarrow x_+$. As $\xi(x)=x\sqrt{-\mu(0)}+\mathcal{O}(x^2)$, we indeed have a non-integrable singularity in the term
\begin{equation*}
  \frac{\mu}{\xi^2}\sim \frac{1}{s^2}.
\end{equation*}
However, as we show in Lemma \ref{integral}, the integral $\lim_{\delta \to 0_+}\int_\delta^{\delta + \rho}G(x,s)/(\xi(s))^2 ds$ can be shown to be bounded for a finite positive $\rho$ using asymptotic relations for $u,v$ and $\xi$ around $0$ (note that both weight functions are continuous and nonzero in the neighbourhood of $0$ and thus do not need to be taken into account here).

Summarising, we have shown that with $\mu^*=\max_{[0,L]}\mu<0$ we have
\begin{equation*}
	\begin{aligned}
	\Vert A_{12}\Vert \leq &\int_0^\rho \left| T(s) \frac{ u_{12}(x) v_{12}(s) - v_{12}(x) u_{12}(s)}{w_2^2(s)} \frac{1}{w_1^2(s)} \right| ds + \\
						+ & \int_\rho^L \left| T(s) \frac{ u_{12}(x) v_{12}(s) - v_{12}(x) u_{12}(s)}{w_2^2(s)} \frac{1}{w_1^2(s)} \right| ds \\
 					  \leq & \frac{1}{\sqrt{-\mu^*}} \int_0^\rho \left| T(s) \frac{ u_{12}(x) v_{12}(s) - v_{12}(x) u_{12}(s)}{w_2^2(s)} \right|  ds + \\
 						+ & \frac{1}{\sqrt{-\mu^*}} \int_\rho^L \left| T(s)\right| \left| \frac{ u_{12}(x) v_{12}(s) - v_{12}(x) u_{12}(s)}{w_2^2(s)}\right| ds  \\
  					  \leq & \frac{1}{\sqrt{-\mu^*(\varepsilon)}} \left(K_1+K_2 \right),
	\end{aligned}
\end{equation*}
where $K_1, K_2$ are constants with respect to both $x$ and $\varepsilon$ while $-\mu^*(\varepsilon) \rightarrow \infty$ as $\varepsilon\rightarrow 0^+$. This completes the proof of the proposition that $\Vert A_{12} \Vert_\infty$ vanishes in the limit $\varepsilon \rightarrow 0$ and hence of the entire approximation. The result is thus what we expected: the smaller the parameter $\varepsilon$, the better description is provided by the asymptotic method. As we shall see in the next section, the proof was to a big extent based on the spectral properties of a Turing system. This suggests that the asymptotic approach is justified in such a setting.

Let us now sum up what we have just proved. In the following theorem, we will assume that the quantities $ w_1(x), w_2(x), u(x)$ and $\mathbf{p}$ are defined as above.

\begin{theorem}\label{theorem1}
Let $\mathbb{Q}(x)$ be a diagonalizable matrix with $\sigma(\mathbb{Q}(x)) \subset (-\infty,0)$ for all $x\in\overline{\Omega}$ with $\Omega  = (0,L)$. Let the eigenvalues $\mu(x)$ of $\mathbb{Q}(x)$ satisfy $\mu \in \mathcal{C}^2(\overline{\Omega})$. Then the scalar function $y$ given by the implicit relation
$$y(x) = u(x) + \int_x^L G(x,s) T(s)y(s)ds, $$
which defines the solution to $\varepsilon^2 \mathbf{y}'' + \mathbb{Q}(x) \mathbf{y} = 0$ as $\mathbf{y}(x) = y(x) \mathbf{p}(x), $ satisfies 
\begin{equation*}
\lim_{\varepsilon \rightarrow 0} \Vert y_{12} - u_{12} \Vert_{L^\infty(\Omega)} = 0,
\end{equation*}
where $\psi_{12}(x) \equiv w_1(x) w_2(x) \psi(x)$.

\end{theorem}

We have thus proved a WKBJ approximation theorem in one spatial variable $x$ for a system of general dimension. The theorem assumes the case of the two modes in the approximation \eqref{1st-order WKBJ} being exponential. Since we were looking for a regular solution and noted that $v$ would become unbounded in the limit $\varepsilon \rightarrow 0$, we chose the decaying mode and approximated the solution with $u$.

In short, one can form an arbitrarily close solution to the problem of interest from two independent Airy functions -- solutions to a very particular problem with linear spatial dependence in reaction kinetics.

\subsection{Oscillatory case}\label{sect_oscillatory}

The natural question is whether this result applies to the case of oscillatory behaviour as well. In fact, we could have answered both these questions at the same time using a more complicated notation. However, for the sake of clarity, we have left the case of oscillatory modes untouched so far. Let us therefore analyze the case of positive eigenvalues. First of all, let us note that the case $\mu(x_0)=0$\footnote{We again refrain from indexing the eigenvalues of $\tilde{\mathbb{Q}}$.} presents additional difficulty in terms of regularity of the solution as the pre-factor $\sqrt{\xi / \xi'}$ is then ill-defined at $x_0$. On the other hand, the corresponding eigenvector $\mathbf{p}(x_0)$ of the matrix of $\mathbb{Q}(x_0)$ (with zero eigenvalue) will exist so that the solution \eqref{solution SR 2D} could still be defined if the regularity issues could be resolved. This is indeed possible by means of additional assumptions on $\xi$, and hence, on $\mathbb{Q}$. The necessary condition is clearly $\xi(x_0) = 0$. Further assumptions could guarantee that $\xi$ will vanish sufficiently quickly at $x_0$. Hence, it is possible to define a regular solution even if we admit eigenvalues crossing $0$. However, this is beyond our interest and the scope of this article so we will not go into much detail here. 

Let us therefore assume $\mu > 0$ throughout $[0,L]$. Then the argument of $K_{1/3}, I_{1/3}$ becomes purely imaginary. As we observe in remark \ref{imaginary xi} in Appendix \ref{details}, using a suitable phase constant, the modified Bessel functions $K_\alpha, I_\alpha$ can still be defined to be real-valued and regular for a purely imaginary argument, while also preserving other properties, such as the recurrent relations \eqref{recurrent}. In fact, the transition to a purely imaginary argument corresponds to replacing the modified Bessel functions with solutions to the ordinary Bessel equation, as the modified equation is obtained from the original Bessel equation using the transformation $\bar{x} = i x$. Let us therefore define $\bar{\xi}(x) = \int_0^x \sqrt{\mu(t)} dt \in \mathbb{R}$ and
\begin{equation}\label{u,v bar}
\begin{aligned}
\bar{u}(x) & \equiv \sqrt{\frac{\bar{\xi}(x)}{\bar{\xi}'(x)}} K_{1/3}(-i \bar{\xi}(x)), \\
\bar{v}(x) & \equiv \sqrt{\frac{\bar{\xi}(x)}{\bar{\xi}'(x)}} I_{1/3}(-i \bar{\xi}(x)).
\end{aligned}
\end{equation}
As we argue in the aforementioned remark \ref{imaginary xi}, these functions satisfy a relation analogous to \eqref{2nd der}, namely
\begin{equation*}\label{2nd der bars}
\psi'' = \left[ - (\bar{\xi}')^2 - T \right] \psi = \left[ -\mu - T \right] \psi,
\end{equation*}
and hence define solutions to equation \eqref{SR 2D} in the case of a positive spectrum of $Q$ in the same manner as $u$ and $v$ do for a negative spectrum, namely by
\begin{equation*}
\bar{\mathbf{y}}(x) = \mathbf{p}(x) \bar{y}(x),
\end{equation*}
with $\bar{y}$ now given as
\begin{equation*}
\bar{y}(x) = a \bar{u}(x) + b \bar{v}(x) + \int_x^L \bar{G}(x,s) T(s)y(s)ds.
\end{equation*}
It can be shown in the same manner as before that this function satisfies 
$$\bar{y}''(x) = -(\bar{\xi}')^2(x) \bar{y}(x) = -\mu(x) \bar{y}(x).$$ 
Therefore, to complete the proof of the approximation theorem for oscillatory modes, we again need to prove that the correction $\bar{y}(x)- a \bar{u}(x) - b \bar{v}(x) $ vanishes in the limit $\varepsilon \rightarrow 0$. Reusing the Neumann series argument, we now arrive at
$$\Vert \bar{y}- a \bar{u} - b \bar{v} \Vert_{L^\infty(\Omega)} \leq \sum_{n=1}^\infty \Vert \bar{A} \Vert^n_\infty \Vert a \bar{u} + b \bar{v} \Vert_{L^\infty(\Omega)} = \frac{\Vert \bar{A} \Vert_\infty}{1 - \Vert \bar{A} \Vert_\infty}  \Vert a \bar{u} + b \bar{v} \Vert_{L^\infty(\Omega)},$$
with
$$\bar{A} f(x) = \int_x^L \bar{G}(x,s) T(s)f(s)ds$$
and $\bar{G}(x,s) = \bar{u}(x)\bar{v}(s) - \bar{u}(s)\bar{v}(x)$. However, the remaining argument that \\ \mbox{ $\lim_{\varepsilon \rightarrow 0} \Vert \bar{A} \Vert_\infty = 0$} simplifies substantially compared to the case of the exponential approximations $u,v$, since now both of the functions $\bar{u}, \bar{v}$ are bounded with respect to $x$. As the exponential behaviour of $u,v$ is now replaced by oscillatory behaviour, estimates of the same form as \eqref{upper bounds} can be obtained for $\bar{u}, \bar{v}$ in terms of $\bar{w}_1(x) = w_1(x) = \mu(x)^{1/4}, \bar{w}_2(x) = 1$. Note that this also means that $\bar{u}, \bar{v}$ are both $\mathcal{O}(\varepsilon^{1/2}), \varepsilon \rightarrow 0$, as follows from the fact that $\mu \sim 1/\varepsilon^2, \varepsilon \rightarrow 0$. Since $T$ is again $\varepsilon$-independent (as the powers of $\varepsilon$ in the nominator and the denominator cancel for each term), this proves that $\Vert \bar{A} \Vert_\infty$ vanishes in the limit $\varepsilon \rightarrow 0$, if the integral can be shown to be bounded with respect to $x$ for a fixed value of $\varepsilon$. However, the asymptotic behaviour of $\bar{u}, \bar{v}$ is the same as that of $u,v$ as $x$ approaches $0$, as can be seen from the relations \footnote{For integer values of $\alpha$ we take the limit.}
\begin{equation}\label{I alpha}
\begin{aligned}
I_\alpha(z) & = (-i)^\alpha J_\alpha (iz), \quad -\pi < \arg z \leq  \frac{\pi}{2},\\
K_\alpha(z)& = \frac{\pi}{2} \frac{I_{-\alpha}(z) - I_\alpha (z)}{\sin (\alpha \pi)}
\end{aligned}
\end{equation}
for $z = -i \bar{\xi}$, with $J_\alpha$ denoting the Bessel function of the first kind. All of these functions are real-valued for both a real and a purely imaginary argument by the choice of the constant $(-i)^\alpha.$ Since we have $J_\alpha(z) \sim z^\alpha$ for $z \rightarrow 0$ as long as $\alpha$ is not a negative integer, lemma \ref{integral} remains true for $\bar{u}, \bar{v},$ just as it was for $u,v$. It follows that the integral is bounded in $x$ on $[0,L]$ and that
$$ \lim_{\varepsilon \rightarrow 0} \Vert \bar{A} \Vert_\infty = 0.$$
This completes the proof of the following variation of theorem \ref{theorem1} for oscillatory approximations:

\begin{theorem}\label{theorem2}
Let $\mathbb{Q}(x)$ be a diagonalizable matrix with $\sigma(\mathbb{Q}(x)) \subset (0, +\infty)$ for all $x\in\overline{\Omega}$ with $\Omega  = (0,L)$. Let the eigenvalues $\mu(x)$ of $\mathbb{Q}(x)$ satisfy $\mu \in \mathcal{C}^2(\overline{\Omega})$.  Then the scalar function $\bar{y}$ given by the implicit relation
$$\bar{y}(x) = a \bar{u}(x) + b \bar{v}(x) + \int_x^L \bar{G}(x,s) T(s)\bar{y}(s)ds, $$
with $a,b$ determined from initial or boundary conditions, which defines the solution to $\varepsilon^2 \mathbf{y}'' + {Q}(x) \mathbf{y} = 0$ as $\mathbf{y}(x) = \mathbf{p}(x) y(x)$, satisfies 
\begin{equation*}
\lim_{\varepsilon \rightarrow 0} \Vert {y} - a \bar{u} - b \bar{v} \Vert_{L^\infty(\Omega)} = 0.
\end{equation*}

\end{theorem}

It should be noted that it is not necessary for the spectrum of $\mathbb{Q}(x)$ to lie entirely on the negative/positive semi-axis. Indeed, combining theorems \ref{theorem1} and \ref{theorem2} and unifying the definitions of $\xi$ and $\bar{\xi}$, $u,v$ and $\bar{u}, \bar{v}$ etc. based on the sign of the eigenvalue, we could now admit $\sigma(\mathbb{Q}(x)) \subset \mathbb{R}$, only excluding the pathological case $\sigma(\mathbb{Q}) = \lbrace 0 \rbrace$. Any nonzero eigenvalue $\mu$ of $\mathbb{Q}$ clearly defines a solution of the form \eqref{solution SR 2D} with $y$ given by \eqref{int solution}. Moreover, a (constant) eigenvalue $\mu = 0$ defines a solution to \eqref{SR 2D} of the form $\mathbf{y}(x) = \mathbf{p}_0(x)(c_1 x +c_2),$ with $\mathbf{p}_0(x)$ denoting the corresponding eigenvector. As noted above, the case of an eigenvalue $\mu(x)$ attaining both zero and nonzero values requires additional regularity analysis (and assumptions) that surpasses the scope and the goals of this text. It is clear, though, that for a matrix with both positive and negative eigenvalues, the WKBJ approximation will consist of both exponential and oscillatory solutions.

\subsection{Complex spectrum}
Having established the notation of sections \ref{sect_exponential} and \ref{sect_oscillatory}, the transition to the case of a general complex spectrum $\sigma \left( \mathbb{Q} \right)$ of $\mathbb{Q}(x)$ is quite straightforward. Once we have established the notions introduced in the previous two sections for a complex spectrum, the approximation theorem will immediately follow from the properties proved therein.

Let us therefore start by defining the square root of a complex number $z = |z| e^{i \varphi}, \varphi \in (- \pi, \pi)$\footnote{We deliberately neglect the left semi-axis of the real axis for reasons clarified later in the text. } by
\begin{equation}\label{complex sqrt}
z^{1/2} \equiv \sqrt{|z|} e^{i \varphi/2}. 
\end{equation}
It follows that 
$$ \left| \arg \left( z^{1/2} \right) \right| < \frac{\pi}{2}, \forall z \in \mathbb{C}\setminus \mathbb{R_-} \Leftrightarrow \Re z^{1/2} > 0, \forall z \in \mathbb{C}\setminus \mathbb{R_-},$$
an observation that we will find very useful soon. Note as well that this definition coincides with the classical square root for $\mu >0$ (even for any $\mu \in \mathbb{R}$ if we define the polar angle on the negative semi-axis to be $\pi$). Throughout this section, it is going to be beneficial to work with the polar coordinate system in the complex plane. Let us therefore assume $\mu(x)$ to be an eigenvalue of $\mathbb{Q}(x)$ and write its negatively taken value in the polar form as $-\mu(x) = |-\mu(x)| e^{i \varphi(x)}=|\mu(x)| e^{i \varphi(x)}.$ Proceeding in a manner analogous to the previous sections, we define
\begin{equation}\label{xi complex} 
\xi(x) \equiv \int_0^x (-\mu)^{1/2}(t)dt = \int_0^x \Re \left( (-\mu)^{1/2}(t) \right) dt  + i \int_0^x \Im \left( (-\mu)^{1/2}(t) \right) dt. 
\end{equation}

\begin{lemma}
Let the eigenvalue $\mu(x)$ of $\mathbb{Q}(x)$ satisfy $\mu(x) \in \mathbb{C}\setminus \mathbb{R_+}, \forall x \in [0,L]. $  Then the function $\xi(x) = \int_0^x (-\mu)^{1/2}(t)dt$ satisfies $\Re \left( \xi(x) \right) > 0, \left( \xi'(x) \right)^2 = - \mu(x).$
\end{lemma}

\begin{proof}
The fact that $\Re(\xi(x)) > 0$ follows immediately from the observation that $\Re( (-\mu)^{1/2}(t) ) > 0, \forall t \in [0,L].$ Furthermore, using the above (polar) form of $-\mu$ and the definition of $z^{1/2},$ we can write 
\begin{equation*}\label{mu^(1/2)}
(-\mu)^{1/2}(x) = \sqrt{|\mu(x)|} \left[ \cos \left( \frac{\varphi(x)}{2} \right) + i \sin \left( \frac{\varphi(x)}{2} \right) \right]. 
\end{equation*}
Differentiating equation \eqref{xi complex} we then have 
\begin{equation*}
	\begin{aligned}
\left( \xi'(x) \right)^2 & = \left[ \Re \left( (-\mu)^{1/2}(x) \right)  + i \Im \left( (-\mu)^{1/2}(x) \right) \right]^2 = \left[ \sqrt{|\mu(x)|} \cos \left( \frac{\varphi(x)}{2} \right) +  i \sqrt{|\mu(x)|} \sin \left( \frac{\varphi(x)}{2} \right) \right]^2 \\
 & = |\mu(x)| \left( \cos  \varphi(x) + i \sin \varphi(x) \right) = - \mu(x)
	\end{aligned}
\end{equation*}
by De Moivre's formula.
\end{proof}
This property allows us again to proceed as before, defining
\begin{equation}\label{u,v complex}
\begin{aligned}
u(x) & \equiv \left( \frac{\xi(x)}{\xi'(x)} \right)^{1/2} K_{1/3} \left( \xi(x) \right), \\
v(x) & \equiv \left( \frac{\xi(x)}{\xi'(x)} \right)^{1/2} I_{1/3}( \xi(x)).
\end{aligned}
\end{equation}
Note that this definition not only (formally) coincides with the definition of $u,v$ given in equation \eqref{u and v}, it also (practically) coincides with the definition of $\bar{u}, \bar{v}$ given by \eqref{u,v bar}. This follows from
\begin{equation}\label{conjugate}
I_\alpha \left( \bar{z} \right) = \overline{ I_\alpha (z)}, \quad K_\alpha \left( \bar{z} \right) = \overline{ K_\alpha (z)},
\end{equation}
and the fact that there we defined $\xi$ without the sign in the square root to make it real and rather opted to work with imaginary units in the argument of $K_\alpha$ and $I_\alpha$. Instead, we could have defined $\xi$ to be purely imaginary (using the sign in the square root) and use definitions of $u$ and $v$ identical to \eqref{u and v} and \eqref{u,v complex}, as both $K_\alpha$ and $I_\alpha$ are real-valued on the real as well as on the imaginary axis (for all positive values of $\alpha$). \cite{Abramowitz} 

It is a straightforward task to show that the complex square root defined in \eqref{complex sqrt} has properties analogous to the real square root, e.g. 
$$ \left( z^{1/2}(x) \right)' = \frac{1}{2 z^{1/2}(x)} z'(x).$$ 
Hence, the properties of the derivatives of $u$ and $v$, including the recurrent relations \eqref{recurrent}, remain identical to those presented in the previous section. There are, however, issues that need to be addressed when considering complex eigenvalues $\mu$. First of all, it is the fact that now the functions $u$ and $v$ are not necessarily real. This is not a substantial problem if $\mathbb{Q}$ is a complex matrix. On the other hand, if $\mathbb{Q}$ is real, then for any eigenvalue $\mu$, its complex conjugate $\overline{\mu}$ is also an eigenvalue of $\mathbb{Q}$. It is thus useful to observe that if $\mu$ defines $\xi$ by equation \eqref{xi complex} and $u,v$ by \eqref{u,v complex}, its complex conjugate  $\bar{\mu}$ generates the complex conjugate of $\xi$, which, in turn, generates the complex conjugate of $u,v$ by the properties \eqref{conjugate} and by the observation $(\bar{z})^{1/2} = \overline{z^{1/2}}.$ A choice of real-valued (combinations of) solutions is thus always possible for real-valued matrices. Note that this property - that the conjugate eigenvalue generates the conjugate of the function - also applies to $G(x,s)$ and $T(x)$ as given in section \ref{sect_exponential} by the properties of complex conjugation with respect to multiplication and reciprocal. 

The ability to obtain real-valued solutions to a real-valued problem allows us to make good sense of the scalar solutions of the form \eqref{int solution} and again define solutions to the vectorial problem by \eqref{solution SR 2D}. We can then proceed as in the previous two sections and again arrive at the necessity to show that the correction to the solution given by the integral in \eqref{int solution} vanishes in the limit $\varepsilon \rightarrow 0.$ This can be done similarly as before but another feature of complex eigenvalues needs to be handled first. For $\mu$ non-real, the weight functions $w_i$ defined by \eqref{weights} may not be real-valued. At this moment, the observation that $\Re(\xi(x)) > 0$ or, equivalently, $\left| \arg (\xi)\right| < \frac{\pi}{2},$ becomes very useful. Due to this property, and again exploiting asymptotic relations 9.7.1 and 9.7.2 from the book \cite{Abramowitz}, it suffices to take the magnitudes of the weight functions $w_i$, \i.e. to define
$$w_1^c(x) \equiv \left| (-\mu(x))^{1/4} \right|, \quad w_2^c(x) \equiv \left| \exp \left( \int_0^x \sqrt{-\mu(s)}ds \right) \right|.$$
Then relations analogous to \eqref{upper bounds} are valid for $w_i$ replaced by $w_i^c$. It should be noted that, unlike in the previous section, using the weight functions $w_i^c$ and changing the measure of the interval cannot be avoided here. The reason is that now the argument of $K_\alpha, I_\alpha$ is generally not purely imaginary and the solutions will hence display a combination of exponential and oscillatory behaviour. This can again be seen from the aforementioned asymptotic relations for $K_\alpha, I_\alpha.$ Hence, for the sake of its regularity, we will again omit the exponentially growing solution $v$ from our approximation, only preserving the exponentially decaying mode $u$.

As the remaining properties of the solutions $u,v$ (including the asymptotic behaviour for $\left| z \right| \rightarrow 0$, which again implies the validity of Lemma \ref{integral}) remain unchanged, the completion of the proof of the following generalization of Theorem \ref{theorem1} is just a matter of repeating the remaining steps in its proof from section \ref{sect_exponential}. The meaning of the quantities $\mathbf{p}, G$ and $T$ again corresponds to their meaning in Theorem \ref{theorem1}.

\begin{theorem}\label{theorem complex}
Let $\mathbb{Q}(x)$ be a diagonalizable matrix and $\Omega  = (0,L)$. Let the eigenvalue $\mu(x)$ of $\mathbb{Q}(x)$ satisfy
\begin{enumerate}
	\item $\mu \in \mathcal{C}^2 \left( \overline{\Omega} \right),$ 
	\item $\left( \forall x \in\overline{\Omega} \right) (\mu(x) \in \mathbb{C}\setminus \mathbb{R_+}).$
\end{enumerate}
Then the scalar function $y$ given by the implicit relation
$$y(x) = u(x) + \int_x^L G(x,s) T(s)y(s)ds, $$
which defines the solution to $\varepsilon^2 \mathbf{y}'' + \mathbb{Q}(x) \mathbf{y} = 0$ as $\mathbf{y}(x) = y(x) \mathbf{p}(x), $ satisfies 
\begin{equation*}
\lim_{\varepsilon \rightarrow 0} \Vert y_{12} - u_{12} \Vert_{L^\infty(\Omega)} = 0,
\end{equation*}
where $\psi_{12}(x) \equiv w_1^c(x) w_2^c(x) \psi(x)$.

\end{theorem}

The case of a positive eigenvalue is analyzed in section \ref{sect_oscillatory}. There are multiple reasons to exclude it here. Firstly, positive eigenvalues lead to qualitatively different behaviour, as then the solution consists of two (bounded) oscillating modes. Secondly, a transition of the eigenvalue $\mu$ through the positive real semi-axis poses a theoretical difficulty as well in that then the polar angle $\varphi$ of $-\mu$ becomes discontinuous. That means that the derivatives of $\mu$ and the second derivatives of $\xi$, which appear in different quantities (e.g. in $T$) and steps of the proof, might not exist. This cannot be simply resolved by allowing $\left| \arg z \right|$ to exceed $\pi$ since then we might lose the validity of the asymptotic expansion for $I_\alpha,$ and hence the upper bounds for $\left| u \right|$ and $\left| v \right|$ given in terms of $w_1^c, w_2^c.$ It hence appears, at least from this theoretical analysis, that the case of a positive spectrum/eigenvalue studied in section \ref{sect_oscillatory} is distinct from the case of a non-positive eigenvalue analyzed in decent generality in this section.


\section{The WKBJ approximation in a reaction-diffusion setting}\label{sect_WKBJ RD}

In this section, we are going to discuss the relation of the problem \eqref{WKBepsilon} to reaction-diffusion equations and the validity of the approximation theorem given in the previous section for these equations and the closely related Turing model for pattern formation. For the most part, our approach here will be strictly mathematical; a more general introduction to reaction-diffusion (RD) systems is given in Appendix \ref{App Turing}. 


Let us start by recalling the form of the linearized RD equations (see equation \eqref{RD lin App} in Appendix \ref{App Turing})
\begin{equation}\label{RD lin}
\textbf{w}_t = \mathbb{D} \Delta \textbf{w} + \mathbb{J}(x) \mathbf{w}.
\end{equation}
By linearity, we expect exponential dynamics and thus apply the ansatz  $\mathbf{w}(t,x, \lambda) \sim e^{\lambda t} \mathbf{y}(x, \lambda)$. Note that in line with the asymptotic approximation \eqref{1st-order WKBJ}, in a heterogeneous domain we admit explicit dependence of the spatial modes (represented by $\mathbf{y}$) on $\lambda$. Substituting this ansatz back into the equations, we obtain a system of the form \eqref{SR 2D} with 
\begin{equation}\label{Q Turing}
\mathbb{Q}(x) = \mathbb{D}^{-1} \left( \mathbb{J}(x) - \lambda \mathbb{I} \ \right),
\end{equation} 
where $\mathbb{J}(x)$ is the Jacobian matrix of the reaction kinetics. It has been shown \cite{Klika} that, as in the classical Turing instability, the relevant growth rates (with $\Re(\lambda) > 0$) are real in the case of two morphogens. Therefore, if one considers the classical two-species system, the assumption on the spectrum of $\mathbb{Q}$ is equivalent (as we are seeking instability, \i.e. positive growth with $\Re(\lambda)>0$) to requiring the matrix of the linearised kinetics $\mathbb{J}(x)$ to have a negative spectrum. This is consistent with the requirement of a stable steady state in Turing instability, but not necessary (a complex pair of eigenvalues in the left half-plane is sufficient). 

We shall use the asymptotic results from the previous section to show two properties of the solution: (i) in the case of time-scale separation, show the previously reported \cite{Klika} eigenmodes of this spatially heterogeneous system including the estimate of the error in making such an approximation (and generalising it to $n$ coupled RD equations); (ii) the fast growing modes always follow these asymptotic results.

The former immediately follows from the previous section when time-scale separation between diffusion and kinetic processes is possible, \i.e. when the parameter $\varepsilon^2 = \frac{D_1 T}{L^2}$, \footnote{Here $D_1$ is a diffusion coefficient, $T$ is a characteristic kinetic timescale, and $L$ is the domain length. A brief overview of the role of these quantities in a RD model is given in Appendix \ref{App Turing}.} given as a ratio of the characteristic temporal scales of the kinetics and that of the diffusion,  is small. An experimental analysis of nodal and lefty \cite{nodal}, the flagship experimental morphogen pair for mammalian pattern formation via a reaction-diffusion mechanism, reveals that the value of the parameter lies in the region $\varepsilon^2 \sim 5.6 \times 10^{-4}$, thus justifying the asymptotic approach with respect to $\varepsilon$ that Krause et al. adopted\cite{Klika}. Then the fact that the eigenmodes of the spatial operator grow with the rate $\lambda$ follows from equation \eqref{SR 2D} with $\tilde{\mathbb{Q}}$
given by \eqref{Q Turing}.

The latter requires a bit more discussion which is why we are henceforth going to focus on the limit of large growth rates $\lambda \rightarrow +\infty$. Let us therefore discuss the properties of $Q$ in this limit, \i.e. considering growth rates $\lambda$ larger than any temporal time-scale encountered.

\begin{proposition} 
For a $n$-component reaction-diffusion system, the eigenvalues $\mu_j(x)$ of $\mathbb{Q}(x)$ satisfy
$$\mu_j(x) { \rightarrow} - \infty \text{ for } { \lambda \rightarrow + \infty}, \forall x \in [0,L],$$
or more precisely $\mu_j(x) \sim -(1/D_j) \lambda \text{ for } \lambda \rightarrow +\infty$.
\end{proposition}

\begin{proof}

  The eigenvalues of $\mathbb{Q}(x)=\mathbb{D}^{-1}(\mathbb{J}(x)-\lambda\mathbb{I})$ are given as the roots of the polynomial
  \begin{equation*}
    0=\det (\mathbb{Q}-\mu \mathbb{I})=\det(\lambda \mathbb{D}^{-1}) \det (\frac{1}{\lambda}\mathbb{J}-\mathbb{I}-\frac{\mu}{\lambda} \mathbb{D}),
  \end{equation*}
  or, equivalently,
  \begin{equation*}
    0=\det [\mathbb{I}+\frac{1}{\lambda}(\mu\mathbb{D}- \mathbb{J})]\\
     = 1+\frac{1}{\lambda}\left(\mu\mathrm{tr} \mathbb{D} - \mathrm{tr} \mathbb{J}\right) + \mathcal{O}\left(\lambda^{-2}\right).
   \end{equation*}
   As $\mathbb{D} $ and $ \mathbb{J}$ are independent of $\lambda$, it follows that there is no solution for $\mu=\mathcal{O}(1)$ for $\lambda\rightarrow +\infty$. Similarly, there is no solution for $\mu\gg\mathcal{O}(\lambda)$. Finally, there are $n$ roots of the order $\lambda$ (which exhausts the spectrum). Denoting $\kappa= \mu / \lambda =\mathcal{O}(1)$, these roots satisfy
   \begin{equation*}
     0=\det \left[\mathbb{I}+\kappa\mathbb{D}- \frac{1}{\lambda} \mathbb{J} \right]=\det [\mathbb{I}+\kappa\mathbb{D}] + \frac{1}{\lambda} \tilde{f}_1(x)+\mathcal{O}\left(\lambda^{-2}\right),
   \end{equation*}
   and hence $\mu_j = - \frac{1}{D_j}\lambda + f_1(x) + \mathcal{O} \left(\lambda^{-1}\right)$. 
\end{proof}

Note that this result is fairly intuitive since we can view $\lambda$ as shifting the spectrum of $\mathbb{D}^{-1}\mathbb{J}(x)$ to the left (although not 'uniformly' due to the fact that we subtract $\lambda \mathbb{D}^{-1}$ rather than $\lambda \mathbb{I}$). 
In complete analogy to section \ref{system}, we could now go on to define $\xi(x)$ by \eqref{xi}, $u(x)$ and $v(x)$ by \eqref{u and v}, the function $y(x)$ by \eqref{int solution} and finally the solution to the vectorial problem using the eigenvector $\mathbf{p}$, $\mathbf{y}=y\mathbf{p}$. Again, we wish to show that the implicit solution $\mathbf{y}$ is well approximated by the Airy functions for large growth rates $\lambda$.  

To reproduce the results from Section \ref{system}, we need to verify certain properties in the limit of large growth rates.
The previous proposition guarantees that the key property 
\begin{equation*}
\mu(x) \sim - C \lambda \text{ for } { \lambda \rightarrow + \infty}, \forall x \in [0,L],
\end{equation*}
with $C > 0,$ is preserved. Note that it is not the linearity in $\lambda$ that is crucial. By substitution, we could allow for any fixed polynomial dependence on an unbounded parameter. It is the fact that the eigenvalue diverges in the same manner for every $x$ in the limit $\lambda \rightarrow + \infty$ that implies the boundedness of the function $T$ given by \eqref{T} with respect to $\lambda$, as we need to guarantee the boundedness of the ratio $\frac{\mu}{\xi^2}$. In this context, it is worth mentioning that under these assumptions, the role of $\lambda$ exactly corresponds to the role of $1/\varepsilon^2$ in the previous section. 

As the properties of the modified Bessel functions $u$ and $v$ remain unchanged, the only step left in the proof of an approximation theorem analogous to \ref{theorem1} is to show that the integral operator defining the solution via \eqref{int solution} is bounded with respect to both $x$ and $\lambda$. The analysis of boundedness in $x$ will be no different from the case in section \ref{system}. Boundedness in $\lambda$ follows from the boundedness 
of the derivatives of the eigenvalues $\mu$ with respect to arbitrarily large growth rates $\lambda$. 

\begin{lemma}\label{unity}
The functions $\mu'(x), \mu''(x)$ are of order unity for $\lambda \rightarrow +\infty$. 
\end{lemma}

\begin{proof}
This observation directly follows from the explicit form of the eigenvalues as obtained in the proof of the last proposition: $\mu^{(k)}= f_1^{(k)}(x)+\mathcal{O}\left(\lambda^{-1}\right)$ for $k\in\mathbb{N}$. 
\end{proof}

The following observation can be made based on this lemma.
\begin{proposition}
Let $\mu(x)$ be an eigenvalue of $\mathbb{Q}(x) = \mathbb{D}^{-1} \left( \mathbb{J}(x) - \lambda \mathbb{I} \ \right)$ and set $\xi(x) \equiv \int_0^x \sqrt{-\mu(t)}dt$. Then the function
\begin{equation*}
T(x; \lambda) = \frac{1}{4} \left[ -\frac{5}{9} \frac{\mu(x)}{\xi^2(x)} + \frac{\mu''(x)}{\mu(x)} - \frac{5}{4} \frac{\mu'^2(x)}{\mu^2(x)} \right]
\end{equation*}
is bounded on $\Omega_{\lambda_0} \equiv [\delta,L] \times (\lambda_0, +\infty)$ for some $\lambda_0 > 0$ and $\delta>0$. In addition, $T[x;\lambda]$ has an integrable singularity at $x=0$ in the sense of lemma \ref{integral}.
\end{proposition}

\begin{proof}
Let us start by recalling that for $\lambda$ sufficiently large the eigenvalue $\mu$ is strictly negative on the entirety of $[0,L]$ by the properties of $\mathbb{J}(x)$. This fact, together with the continuity of $\mu$ on $[0,L]$ and lemma \ref{unity} prove that the second and the third fraction are bounded on $\Omega_\lambda$. The boundedness of the first fraction with respect to $\lambda$ follows from the fact that both $\mu$ and $\xi^2$ are asymptotically linear in $\lambda$ for $\lambda$ approaching $+\infty$. Integrability with respect to $x$ again follows from the properties of $u, v$ and $\xi$, namely the estimates \eqref{estimate1} and \eqref{estimate2} and lemma \ref{integral}. 
\end{proof}

Using the spectral properties of a typical RD system, we have now replicated all properties employed in section \ref{system}, and hence proved the following theorem in the limit of large growth rates. The quantities $ w_1(x), w_2(x), u(x)$ and $\mathbf{p}$ are defined as in section \ref{system}.
\begin{theorem}\label{theorem RD}
Let $\mathbb{J}(x)$ be a diagonalizable matrix of reaction kinetics with $\sigma(\mathbb{J}(x)) \subset (-\infty,0)$ for all $x\in\overline{\Omega}$ with $\Omega  = (0,L)$. Let the eigenvalues $\mu(x) $ of $\mathbb{Q}(x)$ satisfy $\mu \in \mathcal{C}^2(\overline{\Omega})$. Then the scalar function $y$ given by the implicit relation
$$y(x) = u(x) + \int_x^L G(x,s) T(s)y(s)ds, $$
which defines the solution to $\varepsilon^2 \mathbf{y}'' + \mathbb{Q}(x) \mathbf{y} = 0$ as $\mathbf{y}(x) = y(x) \mathbf{p} $, satisfies 
\begin{equation*}
\lim_{\lambda \rightarrow +\infty} \Vert y_{12} - u_{12} \Vert_{L^\infty(\Omega)} = 0,
\end{equation*}
where $\psi_{12}(x) \equiv w_1(x) w_2(x) \psi(x)$.

\end{theorem}


\section{Discussion}
In an attempt to justify the application of the asymptotic WKBJ analysis to a system of linearly coupled reaction-diffusion equations employed in the literature, most recently by Krause et al. \cite{Klika}, we have offered here theorems validating this approximation for systems of ordinary differential equations. First, we build a tool for WKBJ analysis of a linearly coupled system with an almost arbitrary coupling matrix $\mathbb{Q}$. Then we apply it to the analysis of a reaction-diffusion problem where the spectrum of $\mathbb{Q}$ has been shown to be constrained to the real axis in the case of Turing instability. This, however, is only true in a network of two interacting morphogens. Hence, our results allow a direct extension of the previous findings to a higher number of interacting species.

There are several differences between the approach adopted by Krause et al. \cite{Klika} and ours. Most notably, our approximation was not constructed directly for the modes of the form \eqref{1st-order WKBJ}, but instead used these classical WKBJ modes in the proof as weight functions with asymptotic behaviour corresponding to that of our approximate solutions. Consequently, in the case of oscillatory solutions, unlike the solutions presented in their work, our solutions do not reduce to harmonic eigenmodes, so typical for linearized Turing systems without explicit spatial dependence in the kinetics. Although the character of the Bessel function of the first kind, $J_\alpha(x)$, is indeed reminiscent of a damped harmonic mode, its zeros are not generally periodic and only attain periodicity in the limit $x \rightarrow +\infty$. Hence, neither a proper asymptotic expansion of $J_\alpha(x)$ in terms of harmonic functions, nor the converse are possible. It is therefore conceivable that a more direct approach to proving the WKBJ theory for systems of ODEs is possible.

The possibility of further generalizations of our results, most notably of theorems \ref{theorem1} and \ref{theorem2}, to higher spatial dimensions presents another relevant issue. WKBJ method in higher dimensions is typically used to estimate behaviour in certain distinguished direction, effectively reducing the problem to one dimension. A notable exception is \cite{keller1960asymptotic} whose ideas might be shown to be extendable to similar analysis as the one presented in this article.



\section*{Acknowledgement}
Authors are grateful for support from the Czech Grant Agency, project number 20-22092S.


\appendix

\section{Details of the proof from section \ref{system}} \label{details}
In this appendix we are again going to omit the subscripts $j$ as the specific choice of the eigenvalue has no effect on the validity of the considerations and observations made here.
\begin{lemma}\label{T proof}
The functions $$u(x) = \sqrt{\frac{\xi(x)}{\xi'(x)}} K_{1/3} \left( \xi(x) \right), \qquad
v(x) = \sqrt{\frac{\xi(x)}{\xi'(x)}} I_{1/3} \left( \xi(x) \right) $$
both satisfy the relation
\begin{equation} \label{psi}
\psi'' = \left[ (\xi')^2 - T \right] \psi,
\end{equation}
with
\begin{equation*}
T(x) = \frac{1}{4} \left[ \frac{5}{9} \left( \frac{\xi'(x)}{\xi(x)} \right)^2 + 2 \frac{\xi'''(x)}{\xi'(x)} - 3 \left( \frac{ \xi''(x) }{\xi'(x)} \right)^2 \right].
\end{equation*}
\end{lemma}

\begin{proof}

We are going to offer the proof for $u$, the proof for $v$ is analogous. To simplify the notation, we are going to omit the argument $x$ of $\xi$ and its derivatives. Let us start by expressing the first two derivatives of $u$ with respect to $x$:
\begin{equation*}
\begin{aligned}
u'(x) = & \frac{1}{2} \sqrt{\frac{\xi'}{\xi}} \left( 1 - \frac{\xi \xi''}{\xi'^2}\right) K_{1/3} ( \xi ) + \sqrt{\frac{\xi}{\xi'}} K_{1/3}' ( \xi ) \xi', \\
u''(x)  = & - \frac{1}{4} \left( \frac{\xi'}{\xi} \right)^{3/2} \left( 1 - \frac{\xi \xi''}{\xi'^2}\right)^2 K_{1/3} ( \xi ) - \frac{1}{2} \sqrt{\frac{\xi'}{\xi}} \left( \frac{\xi' \xi'' + \xi \xi'''}{\xi'^2} - \frac{2 \xi \xi' (\xi'')^2}{\xi'^4} \right) K_{1/3} ( \xi ) \\
& + \sqrt{\frac{\xi'}{\xi}} \left( 1 - \frac{\xi \xi''}{\xi'^2}\right) K_{1/3}' ( \xi ) \xi' + \sqrt{\frac{\xi}{\xi'}} K_{1/3}'' ( \xi ) (\xi')^2 + \sqrt{\frac{\xi}{\xi'}} K_{1/3}' ( \xi ) \xi''.
\end{aligned}
\end{equation*} 
Using the recurrent relations \eqref{recurrent} and the symmetric property $K_{-\alpha}(x) = K_\alpha(x)$, we now express
\begin{equation}\label{derivatives}
\begin{aligned}
K_{1/3}'(z) & = \frac{1}{3z}K_{1/3}(z) - K_{4/3}(z), \\
K_{1/3}''(z) & = \frac{9z^2 - 2}{9z^2}K_{1/3}(z) + \frac{1}{z} K_{4/3}(z).
\end{aligned}
\end{equation} 
Substituting back and expressing $u''(x) = c_{1/3}(x)K_{1/3}(\xi) + c_{4/3}(x)K_{4/3}(\xi)$, we find $c_{4/3}(x) = 0$ and
$$c_{1/3}(x) = \sqrt{\frac{\xi}{\xi'}} \left[ -\frac{5}{36} \left( \frac{\xi'}{\xi} \right)^{2} + \frac{3}{4} \left( \frac{\xi''}{\xi'} \right)^{2} - \frac{1}{2} \frac{\xi'''}{\xi'} + (\xi')^2 \right] =  \sqrt{\frac{\xi}{\xi'}} \left[ (\xi')^2 - T \right].$$

\end{proof}

\begin{remark}\label{imaginary xi}
Note that if the eigenvalue in the definition of $\xi$, equation \eqref{xi}, has the opposite (positive) sign (over the entire interval $[0,L]$), the argument $\xi$ of $K_\alpha, I_\alpha$  becomes purely imaginary. However, this has no effect on either the pre-factor $\sqrt{{\xi}/{\xi'}}$ or the recurrent relations \eqref{recurrent}. For more clarity, let us reintroduce $\bar{\xi}(x) = \int_0^x \sqrt{\mu(t)} dt \in \mathbb{R}$ and 
$$\bar{u}(x) \equiv \sqrt{\frac{\bar{\xi}(x)}{\bar{\xi}'(x)}} K_{1/3}(-i \bar{\xi}(x)),$$
$$\bar{v}(x) \equiv \sqrt{\frac{\bar{\xi}(x)}{\bar{\xi}'(x)}} I_{1/3}(-i \bar{\xi}(x))).$$
Unlike $u$ and $v$, these functions are not exponential but rather oscillatory. This can be expressed more precisely recalling the relations \eqref{I alpha}.
for $z = -i \bar{\xi}$. The pre-factor $(-i)^\alpha$ there guarantees the right phase constant so that all of $I_\alpha, J_\alpha, K_\alpha$ are real-valued for real as well as purely imaginary arguments. Differentiating $\bar{u}$ and $\bar{v}$ again as composite functions of $x$ and setting $z = -i \bar{\xi}$ in the relations \eqref{derivatives}, we can again express second derivatives of $\bar{u}$ in terms of $K_{1/3}, K_{4/3}$ (and analogously for $\bar{v}$ using $I_{1/3}, I_{4/3}$). Since the coefficient $\bar{c}_{4/3}(x)$ just obtains a factor of imaginary unit compared to ${c}_{4/3}(x)$, we will this time find that $\bar{u}, \bar{v}$ satisfy
$$\psi''(x) = [- \left( \bar{\xi}'(x) \right)^2 -T(x)] \psi (x) = [-\mu(x) - T(x)] \psi(x). $$
It is worth mentioning that the value of $T(x)$ remains invariant if $\xi$ is replaced by $\pm i \xi$ so there is no necessity to redefine it in terms of $\bar{\xi}$. 

Alternatively - but slightly more vaguely - we could have just argued that 'everything' - except the sign of $(\xi')^2$ - 'remains the same' under the substitution $\xi = -i \bar{\xi}$, (which also means the sign of $\mu$ in the $(\xi')^2$-term in equation \eqref{psi} is the same for $u,v$ and $\bar{u}, \bar{v}$, as may be noted), and hence that lemma \ref{T proof} remains valid up to this sign under this transformation. This observation is very useful for the transition between the exponential solutions and the solutions with oscillatory behaviour.
\end{remark}

\begin{lemma}\label{integral}
Let $\rho > 0$ be a small but finite parameter. With notation as in section \ref{system}, we have
\begin{equation}\label{int}
\lim_{\delta \rightarrow 0_+} \left| \int_\delta^{\delta+\rho } \frac{G(\delta, s)}{\xi^2(s)} ds \right| < + \infty.
\end{equation}

\end{lemma}

\begin{remark*}
We only wish to study the effect of the lower bound in the integral approaching $0$, neglecting any possible effects of $\varepsilon$, as these were discussed in section \ref{system}. Therefore, we will consider $\varepsilon$ to have a fixed value throughout this proof. 
\end{remark*}

\begin{proof}
 Since the integrand in $\xi(s) = \int_0^s \sqrt{-\mu(t)} dt$ is strictly positive by assumption, $\xi$ has a nonzero continuous derivative in $s=0$, and hence, asymptotically satisfies $\xi(s) \sim s, s \rightarrow 0_+$. The asymptotic relations for $K_\alpha, I_\alpha$ can be found in \cite{Abramowitz} or obtained by the method of dominant balance from the modified Bessel equation. Either way, we have $K_\alpha(z) \sim z^{-\alpha}, I_\alpha(z) \sim z^{\alpha}, z \rightarrow 0.$ For $u$ and $v$, this translates into
$$u(t) \sim \left( \xi(t) \right)^{1/6} \sim t^{1/6}, t \rightarrow 0,$$
$$v(t) \sim \left( \xi(t) \right)^{5/6} \sim t^{5/6}, t \rightarrow 0.$$
As integration preserves the relation $\sim$, it follows immediately that for the integral in \eqref{int}, to lowest order in $\delta$, we have
$$\int_\delta^{\delta + \rho} \frac{u(\delta)v(s) - v(\delta) u(s)}{\xi^2(s)} ds \sim c_1 \delta^{1/6} \int_\delta^{\delta + \rho} \frac{s^{5/6}}{s^2}ds - c_2 \delta^{5/6} \int_\delta^{\delta + \rho} \frac{s^{1/6}}{s^2}ds \sim c \delta^0, \delta \rightarrow 0_+.$$
\end{proof}

\section{An introduction to reaction-diffusion equations and Turing instability}\label{App Turing}
As our main motivation for studying the WKBJ method was reaction-diffusion (RD) equations, let us here introduce the basic notions and summarize the basic properties of these models. In general, a RD equation is a partial differential equation of the form
\begin{equation} \label{Turing}
	\frac{\partial \mathbf{c}}{\partial t} = \bar{\mathbb{D}} \Delta \mathbf{c} +  \mathbf{f}(\mathbf{c}) ,
\end{equation}	 
where $\mathbf{c}: \Omega \times \mathbb{R} \rightarrow \mathbb{R}_+^n$, with $\Omega \subset \mathbb{R}^m$, is the concentration vector of the underlying $n$ chemical substances\footnote{In other contexts $\mathbf{c}$ can also represent other quantities, e.g. concentrations of animal species (\i.e. a continuous model of occurrence) in ecological problems.}, $\mathbf{f}:\mathbb{R}^n \times \Omega \rightarrow \mathbb{R}^n$ describes their reaction kinetics, $\Delta$ is the Laplace operator representing diffusion (applied componentwise), and $\bar{\mathbb{D}}$ is the diagonal $n \times n$ matrix of (positive) diffusion coefficients. Note that here, for simplicity, we assumed spatially independent diffusion coefficients as otherwise the diffusion term would have to be replaced by $\nabla \cdot \left( \bar{\mathbb{D}} \nabla \mathbf{c} \right)$, where $\nabla \mathbf{c}$ is a tensor and both $\nabla$ operators are applied accordingly. \footnote{Both $\nabla \mathbf{c}$ and $\bar{\mathbb{D}} \nabla \mathbf{c}$ can be viewed as $n \times m$ matrices. The divergence operator then acts on each line of this matrix separately.} We also assumed homogeneity by omitting an explicit spatial dependence of $\mathbf{f}$; adding an explicit $x$-dependence would mean heterogeneity. Although it turns out that most of the following analysis gives analogous results for these two different cases, we will distinguish between them consistently. For simplicity, our analysis will mostly be aimed at the homogeneous case. The generalizations to heterogeneity are a recent result by Krause et al.\cite{Klika}, who employed asymptotic methods and the WKBJ approximation in their analysis, thus co-motivating our attempt of a formal generalization of this approach to multidimensional systems and RD equations. Let us add that in the context of embryogenesis or - more generally - \textit{spontaneous} pattern formation we usually equip equation \eqref{Turing} with the Neumann (zero-flux) boundary conditions
\begin{equation*}
(\vec{n}\cdot\nabla) \mathbf{c}(\vec{r},t) = 0 \text{ for } \vec{r} \in \partial \Omega,
\end{equation*} 
where $\vec{n}$ represents the outward normal to the boundary $\partial \Omega$. 

A non-dimensionalization procedure with respect to a characteristic reaction time scale $T$, spatial scale $L$ and concentration scale $C$ is usually applied to equation \eqref{Turing} to recover
\begin{equation*} \label{Turing nonD}
\frac{\partial \textbf{c}}{\partial t} \equiv \textbf{c}_t = \varepsilon^2 \mathbb{D} \Delta \mathbf{c} + \mathbf{f}(\mathbf{c}),
\end{equation*}
where we reused $\mathbf{c}, \mathbf{f}$ and $t$ to now represent non-dimensional quantities.  Here we preserved the parameter $\varepsilon^2 = \frac{D_1 T}{L^2}$ explicitly to (shortly) obtain exactly the form of equation that is considered in section \ref{WKBepsilon} asymptotically in the limit $\varepsilon \rightarrow 0$. $D_1$ is usually chosen to be the greatest of the diffusion coefficients, that is the greatest element of the diagonal matrix $\bar{\mathbb{D}}$. We repeat here the observation made in that section, namely that for the morphogen pair nodal and lefty, the parameter value is of order $\varepsilon^2 \sim 10^{-4}$.\cite{nodal} In section \ref{sect_WKBJ RD}, on the other hand, we consider the limit of large growth rates and show that it very much corresponds to the former limit. Hence, there we do not represent $\varepsilon$ explicitly in the equation and just use the form
$$ \textbf{c}_t = \mathbb{D} \Delta \mathbf{c} + \mathbf{f}(\mathbf{c}),$$
which is the form usually used for a general (\i.e. not strictly asymptotic) analysis of reaction-diffusion equations. Either way, for a two-component RD system, the (non-dimensionalized) matrix $\mathbb{D}$ has the form
$$\mathbb{D} = \begin{pmatrix}
1 & 0 \\
0 & d 
\end{pmatrix},$$
with $d = D_2/D_1 < 1$ is the ratio of the diffusion coefficients. We will see shortly that $d=1$ is not permissible if the system is to display Turing instability.

We then proceed by assuming the existence of a homogeneous steady state $\mathbf{c}_*$ satisfying 
\begin{equation*}\label{equil}
\mathbf{f}(\mathbf{c}_*) = 0.
\end{equation*}
Since we wish to apply linear stability analysis, we need to linearize equation \eqref{Turing} around the fixed point $\mathbf{c}_*$. This procedure, introducing the perturbation $\mathbf{w} \equiv \mathbf{c - c}_*$, yields
\begin{equation*}
\textbf{w}_t = \mathbb{D} \Delta \textbf{w} + D \mathbf{f}(\mathbf{w=0})\mathbf{w} + \mathcal{O} \left( \Vert \mathbf{w} \Vert^2 \right).
\end{equation*}
Denoting the Jacobian matrix of the kinetics $D \mathbf{f}(\mathbf{0}) \equiv \mathbb{A}, $ we can see that, upon neglecting the $\mathcal{O} \left( \Vert \mathbf{w} \Vert^2 \right)$-terms, the linear stability of the solution to the linear equation 
\begin{equation}\label{RD lin App}
\textbf{w}_t = \mathbb{D} \Delta \textbf{w} + \mathbb{A} \mathbf{w}
\end{equation}
will depend on the spectrum of the Laplace operator as well as that of $\mathbb{A}$. For $\Omega$ bounded with sufficiently smooth boundary, the Laplace operator will have a purely discrete spectrum in $L^2(\Omega)$, hence allowing us to expand any (possible) solution using its eigenvectors $w_k$ given by
\begin{equation*}
	\begin{aligned}
	\Delta {w}_k & = - k^2 {w}_k, \\
	 \left( \vec{n} \cdot \nabla \right) {w}_k & = 0 \text{ on } \partial \Omega,
	\end{aligned}
\end{equation*}
and discuss the stability of the individual modes. Exploiting linearity by using separation of variables, we obviously obtain exponential dynamics, and the solution to \eqref{RD lin App} will hence be given as
\begin{equation*}
\mathbf{w}(\vec{r}, t) = \sum_{k=0}^\infty \mathbf{v}_k e^{\lambda_kt} w_k(\vec{r})
\end{equation*}
for some vectors $\mathbf{v_k} \in \mathbb{R}^n$ determined by the initial conditions. Its linear stability clearly depends on the signs of (the real part of) the eigenvalue $\lambda_k$. 

For the heterogeneous case (\i.e. for spatially dependent kinetics $\mathbf{f}(\mathbf{c},x)$), we generalize the notion of a homogeneous equilibrium by assuming a \textit{slowly varying} equilibrium $\mathbf{c}_*(x)$, only admitting spatial derivatives of scale $\mathcal{O}(1)$ or smaller (\i.e. excluding derivatives of scale $\mathcal{O}(1/\varepsilon)$ or larger). From linearization about this equilibrium, \i.e. for $w(t,x) = c(t,x) - c_*(x)$, we then obtain
$$\mathbf{w}_t = \varepsilon^2 \mathbb{D} \Delta \mathbf{w} +  \mathbb{J}(x)\mathbf{w},$$
with $\mathbb{J}(x)$ denoting the Jacobian matrix of the map $\mathbf{f}$ evaluated at the equilibrium $\mathbf{c}_*(x)$. Note the analogy between the assumptions on the derivatives of $\mathbf{c}_+(x)$ and the assumptions on the derivatives of the eigenvalues $\mu(x)$ used in sections \ref{WKBepsilon} and \ref{sect_WKBJ RD} as well as the property proved in lemma \ref{unity}. By linearity, the dynamics of this system will again be exponential. However, due to the explicit spatial dependence of the stability matrix $\mathbb{J}(x)$ (and in line with the form of first-order approximation \eqref{1st-order WKBJ}), we allow here for a $\lambda$-dependence of the spatial modes and apply the ansatz $\mathbf{w}(t,x, \lambda) \sim e^{\lambda t} \mathbf{y}(x, \lambda)$. Note that in the homogeneous case, there was no $\lambda$-dependence in the spatial modes $\mathbf{y}_k = \mathbf{v}_k w_k(x)$.

\subsection{Diffusion-driven instability}
Let us now present the idea of \textit{diffusion-driven} (or \textit{Turing}) \textit{instability}, hereinafter sometimes shortened to DDI. As the term reveals, Turing's idea was that diffusion could become the driving force of pattern formation in that it could cause instability of a homogeneous steady state if an inhomogeneous perturbation kicked in. For this to occur (or at least be admissible), several requirements must be fulfilled. Firstly, \textit{stability} of the homogeneous steady state $c_*$ is required. A \textit{homogeneous} state eliminates any effect of diffusion so that the linearized reaction-diffusion equation \eqref{RD lin App} gives $\textbf{w}_t = \mathbb{A} \mathbf{w}$ for the stability (Jacobian) matrix $\mathbb{A} = D\mathbf{f}(\mathbf{c}_*).$ For a two-component RD equation, the eigenvalues of $\mathbb{A}$ are given by
$$\lambda^\pm = \frac{1}{2} \left[ \tr \mathbb{A} \pm \sqrt{(\tr \mathbb{A})^2 - 4\det \mathbb{A} }  \right],$$
so requiring their (real parts') negativeness, and hence stability, is equivalent to requiring both
\begin{equation} \label{stab cond}
\tr \mathbb{A} < 0 \text{ and } \det \mathbb{A}  > 0.
\end{equation}

Assuming that the steady state is indeed stable, we wish to see how diffusion can be the cause of instability and result in the onset of pattern formation. Let us therefore consider an inhomogeneous perturbation of the equilibrium $\mathbf{c}_*$. Then diffusion becomes a factor and we need to investigate the stability of the solutions to the linearized equations. The corresponding eigenvalues are now given by the characteristic equation 
\begin{equation} \label{dif instab}
\det \left( \lambda \mathbb{I} -  \mathbb{A} + k^2 \mathbb{D} \right) = \det \left( \lambda \mathbb{I} - \mathbb{D} \left( \mathbb{D}^{-1} \mathbb{A} - k^2 \mathbb{I} \right) \right) = 0.
\end{equation}
For diffusion-driven instability to occur, it is necessary that for some $k$ we have $\Re \lambda \left( k^2 \right) > 0$. An immediate observation is that the substances must not diffuse at equal rates; if $\mathbb{D}$ was a (positive) multiple of the identity matrix $\mathbb{D} = d \mathbb{I}$, we would just be looking for eigenvalues of $\mathbb{A}$ of the form $\tilde{\lambda} = \lambda + d k^2$. Since all eigenvalues $\tilde{\lambda}$ of $\mathbb{A}$ have negative real parts, so would do $\lambda$ and we would have stability. Hence, the substances (e.g. chemicals or species) need to diffuse at different rates. For further intuition we again study the case of a two-component system. Denoting $\mathbb{B}_k = \mathbb{D} \left( \mathbb{D}^{-1} \mathbb{A} - k^2 \mathbb{I} \right)$, we now obtain the \textit{dispersion relation}, \i.e. the relation between growth rates (and thus frequencies) and wavenumbers (and thus wavelengths), as
$$\lambda_k^\pm = \frac{1}{2} \left[ \tr \mathbb{B}_k \pm \sqrt{(\tr \mathbb{B}_k)^2 - 4\det \mathbb{B}_k }  \right].$$
First, let us note that $\tr \mathbb{B}_k = \tr \mathbb{A} - k^2(1+d) < 0$ by conditions \eqref{stab cond}. Hence, if the growth rate $\lambda_k$ (or its real part) is to be positive for any $k$, we must have  necessarily have $\det \mathbb{B}_k < 0$ so that the absolute value of the square root prevails over the negative trace of $\mathbb{B}_k$. Expressing $\det \mathbb{B}_k$ as a polynomial in $k^2$ whose roots are given by the eigenvalues of $\mathbb{D}^{-1}\mathbb{A}$, we obtain
$$\det \mathbb{B}_k = d \det(\mathbb{D}^{-1}\mathbb{A} - k^2 \mathbb{I}) = d \left[ (k^2)^2 - \tr(\mathbb{D}^{-1}\mathbb{A}) k^2 + \det(\mathbb{D}^{-1}\mathbb{A})  \right] \equiv h(k^2).$$
From $\det(\mathbb{D}^{-1}\mathbb{A}) = \frac{1}{d} \det \mathbb{A}$, the third term is positive by \eqref{stab cond}. If $h(k^2)$ is to be negative anywhere, the second term must be negative, yielding the condition $\tr(\mathbb{D}^{-1}\mathbb{A})>0$. Furthermore, as $h$ is quadratic in $k^2$ with a positive coefficient of the 'quadratic' term $(k^2)^2$, the necessary and sufficient condition for $h$ to attain negativity for some $k$ is that the vertex of the corresponding parabola be negative. This vertex is given by the value of $h$ at $k^2_{min} = \frac{1}{2}\tr(\mathbb{D}^{-1}\mathbb{A}).$ Inserting this value back into $h$, we have
$$h_{min} = h(k^2_{min}) = \det(\mathbb{D}^{-1}\mathbb{A}) - \frac{1}{4} \left( \tr(\mathbb{D}^{-1}\mathbb{A}) \right)^2.$$
The negativity of the vertex $h_{min}$ is hence equivalent to requiring $\left( \tr(\mathbb{D}^{-1}\mathbb{A}) \right)^2 - 4 \det(\mathbb{D}^{-1}\mathbb{A}) > 0 $. For diffusion-driven instability, we thus have the conditions
\begin{equation*}\label{cond all}
\tr  \mathbb{A} < 0, \qquad \det \mathbb{A} >0, \qquad \tr(\mathbb{D}^{-1}\mathbb{A}) > 0, \qquad \left( \tr(\mathbb{D}^{-1}\mathbb{A}) \right)^2 - 4 \det(\mathbb{D}^{-1}\mathbb{A})>0.
\end{equation*}
Let us note that for the heterogeneous case, Krause et al. arrive at exactly these four conditions for DDI with $\mathbb{A}$ replaced by the heterogeneous stability matrix $\mathbb{J}(x)$.\cite{Klika}

We conclude this brief introduction to RD equations and Turing instability by observing that the range of the unstable eigenvalues can be specified precisely for the case of a two-component system: it is clearly given by all the wavenumbers $k^2$ lying between the two roots of the equation $h(k^2) = 0.$ From that, we can specify the unstable modes to be given by precisely those eigenmodes $w_k$ whose wavenumbers satisfy
\begin{equation*} \label{eigenvalues}
	\begin{gathered}
k_{min}^2 \equiv \frac{1}{2} \left[ \tr(\mathbb{D}^{-1}\mathbb{A}) - \sqrt{(\tr(\mathbb{D}^{-1}\mathbb{A}))^2 - 4\det(\mathbb{D}^{-1}\mathbb{A})}\right]  < k^2 \\
< \frac{1}{2} \left[ \tr(\mathbb{D}^{-1}\mathbb{A}) + \sqrt{(\tr(\mathbb{D}^{-1}\mathbb{A}))^2 - 4\det(\mathbb{D}^{-1}\mathbb{A}) }\right] \equiv k_{max}^2.
	\end{gathered}
\end{equation*}

\bibliographystyle{plain}
\bibliography{RD_WKBJ}

\end{document}